\documentclass[aps, pra, reprint, superscriptaddress]{revtex4-2}

\pdfoutput=1

\usepackage{header}

\graphicspath{{graphics/}}

\begin{document}

\bibliographystyle{apsrev4-2}

\title{Optimal Scheduling of Graph States via Path Decompositions}

\author{Samuel J. Elman}
\email{samuel.elman@uts.edu.au}
\affiliation{Centre for Quantum Software and Information, School of Computer Science, Faculty of Engineering \& Information Technology, University of Technology Sydney, NSW 2007, Australia}

\author{Jason Gavriel}
\email{jason@gavriel.au}
\affiliation{Centre for Quantum Software and Information, School of Computer Science, Faculty of Engineering \& Information Technology, University of Technology Sydney, NSW 2007, Australia}
\affiliation{Centre for Quantum Computation and Communication Technology}

\author{Ryan L. Mann}
\email{mail@ryanmann.org}
\homepage{http://www.ryanmann.org}
\affiliation{Centre for Quantum Software and Information, School of Computer Science, Faculty of Engineering \& Information Technology, University of Technology Sydney, NSW 2007, Australia}
\affiliation{Centre for Quantum Computation and Communication Technology}

\begin{abstract}
    We study the optimal scheduling of graph states in measurement-based quantum computation, establishing an equivalence between measurement schedules and path decompositions of graphs. We define the spatial cost of a measurement schedule based on the number of simultaneously active qubits and prove that an optimal measurement schedule corresponds to a path decomposition of minimal width. Our analysis shows that approximating the spatial cost of a graph is $\textsf{NP}$-hard, while for graphs with bounded spatial cost, we establish an efficient algorithm for computing an optimal measurement schedule.
\end{abstract}

\maketitle

\section{Introduction}
\label{section:Introduction}

Measurement-based quantum computation is an approach to quantum computation where adaptive measurements are performed on an initially prepared entangled resource state~\cite{raussendorf2001one, hein2006entanglement}. In this paper, we study the scheduling of measurement-based quantum computation on a class of resource states known as graph states. Specifically, we establish an equivalence between measurement schedules and path decompositions of graphs.

Previous work has studied the optimisation of measurement-based quantum computation by designing graph states specific to the computation~\cite{ferguson2021measurement, vijayan2024compilation, sunami2022graphix, marqversen2023applications, kaldenbach2023mapping, schroeder2023deterministic, ruh2024quantum}. The choice of graph state has a natural associated cost in terms of the number of qubits and entangling gates. However, the entanglement structure of graph states implies that the entire state may not need to be prepared simultaneously~\cite{raussendorf2003measurement}. Consequently, we consider the spatial cost of a given graph state, based on the number of simultaneously active qubits. Specifically, we consider optimisation only in the scheduling of measurements for a fixed graph state.

Our results imply that the spatial cost of a measurement-based quantum computation scales with the pathwidth of the graph. Further, our analysis implies that approximating the spatial cost of a graph is \textsf{NP}-hard in general. For graphs with bounded spatial cost, we establish an efficient algorithm for computing an optimal measurement schedule.

We explore the implications of our results for implementations of fault-tolerant quantum computation. We argue that a low-degree graph, facilitating only nearest neighbour interactions such as the square lattice, is a suitable choice for reducing spatial resources.

This paper is structured as follows. In Section~\ref{section:Framework}, we introduce the necessary framework for our work. Then, in Section~\ref{section:OptimalScheduling}, we prove our main result, which establishes an equivalence between measurement schedules and path decompositions of graphs. Then, in Section~\ref{section:ComputationalComplexity}, we explore the computational complexity of the spatial cost and optimal measurement schedules. In Section~\ref{section:Implementations}, we explore the implications of our results for implementations of fault-tolerant quantum computing. Finally, we conclude in Section~\ref{section:ConclusionAndOutlook} with some remarks and open problems.

\section{Framework}
\label{section:Framework}

In measurement-based quantum computation, a register of qubits is initially prepared in an entangled resource state and quantum operations are executed by performing measurements on individual qubits within this state~\cite{raussendorf2001one, hein2006entanglement}. The focus of this work is on graph states, a type of resource state, which are represented by graphs with vertices corresponding to qubits and the edges denoting entanglement between them. Formally, for a graph $G=(V, E)$, the graph state $\ket{G}$ is defined by
\begin{equation}
    \ket{G} \coloneqq \left(\prod_{e \in E}CZ_e\right)\ket{+}^{\otimes V}. \notag
\end{equation}

Measurements of individual qubits project the graph state into an eigenstate of the measurement observable. When these measurement observables are Pauli matrices, measurements are equivalent to performing Clifford operations. Alternatively, an appropriate choice of measurement observables allows for arbitrary single-qubit operations to be performed. The inherent randomness of quantum measurement necessitates an adaptive approach, meaning the measurement operations are decided based on the outcomes of previous measurements.

A \emph{measurement sequence} is an ordering of the qubits, denoting the sequence in which they are measured. The entanglement structure of graph states implies that the entire state may not need to be prepared simultaneously~\cite{raussendorf2003measurement}. Consequently, qubits can be sequentially initialised and entangled, ensuring only those necessary for the subsequent measurement are prepared. After a qubit has been measured, its physical manifestation can be reinitialised to represent a different vertex in the graph state. We consider a \emph{measurement schedule} that encompasses the initialisation and measurement of qubits within the graph state.
\begin{definition}[Measurement schedule]
    \label{definition:measurementschedule}
    A measurement schedule on a graph state $\ket{G}$ is a sequence of initialising and measuring qubits, such that:
    \begin{enumerate}
        \item[M1.] A qubit is initialised exactly once.
        \item[M2.] A qubit is measured only after all neighbouring qubits are initialised.
    \end{enumerate}
\end{definition}
We say that a qubit is \emph{active} if it has been initialised but not yet measured. We shall represent a measurement schedule as a sequence of subsets of vertices $(X_i)_{i=1}^n$, with each subset corresponding to the qubits that are active at a given step. The \emph{cost} of a measurement schedule is $\max_{i}\abs{X_i}$. The \emph{spatial cost} $\operatorname{sc}(\ket{G})$ of a graph state $\ket{G}$ is the minimum cost over all possible measurement schedules on $\ket{G}$. We say that a measurement schedule is \emph{optimal} if its cost is equal to the spatial cost. Note that there may be multiple optimal measurement schedules for a given graph state.

Since measurement-based quantum computing is capable of universal quantum computation, any quantum computation can be represented as a set of measurements on a particular graph state. Notably, the graph state enabling a specific quantum computation is not unique. The original concept of measurement-based quantum computation involved the use of a universal cluster state, represented by a graph state structured as a square lattice~\cite{raussendorf2001one}. However, it is possible to modify the graph state to minimise the number of qubits. This minimisation is achieved by designing a graph state that is specific to the computation.

There are several methods for designing these specific graph states. One approach involves initially performing all Pauli measurements, resulting in the formation of an alternative graph structure~\cite{sunami2022graphix}. Note that since these are Pauli measurements, the resulting graph can be precomputed, thereby facilitating the application of a measurement schedule. Another approach for designing these specific graph states involves decomposing the computation in terms of magic-state teleportation~\cite{vijayan2024compilation}.

We now define the concept of a \emph{path decomposition} of a graph.
\begin{definition}[Path decomposition]
    \label{definition:pathdecomposition}
    A path decomposition of a graph $G=(V, E)$ is a sequence $(X_i)_{i=1}^n$ of subsets of $V$, such that:
    \begin{enumerate}
        \item[P1.] For each $v \in V$, there exists an $i$ such that $v \in X_i$.
        \item[P2.] For each $e \in E$, there exists an $i$ such that $e \subseteq X_i$.
        \item[P3.] For all $i \leq j \leq k$, if $v \in X_i \cap X_k$, then $v \in X_j$.
    \end{enumerate}
\end{definition}
The \emph{width} of a path decomposition is $\max_{i}\abs{X_i}-1$. The \emph{pathwidth} $\operatorname{pw}(G)$ of a graph $G$ is the minimum width over all possible path decompositions of $G$~\cite{robertson1983graph}. An example of a path decomposition of a graph is given in Fig.~\ref{figure:pathdecompositionexample}.

\begin{figure}[ht!]
    \centering
    \includegraphics[width=0.40\textwidth]{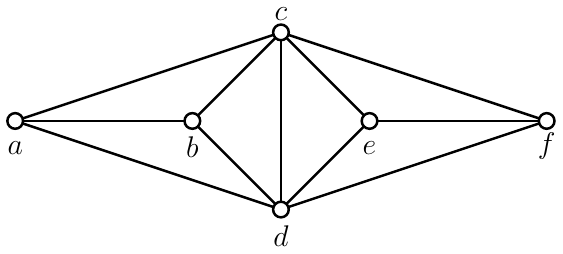}
    \caption{A graph that has a path decomposition given by $(\{a,b,c,d\}, \{c,d,e,f\})$. The width of this decomposition is three.}
    \label{figure:pathdecompositionexample}
\end{figure}

\section{Optimal Scheduling}
\label{section:OptimalScheduling}

In this section, we establish an equivalence between measurement schedules and path decompositions. Consequently, we find that an optimal measurement schedule is determined by a path decomposition of minimal width. Our main result is as follows.
\begin{theorem}
    \label{theorem:PathDecompositionMeasurementScheduleEquivalence}
    Let $G=(V, E)$ be a graph and let $\mathcal{X}=(X_i)_{i=1}^n$ be a sequence of subsets of $V$. The following statements are equivalent:
    \begin{enumerate}
        \item $\mathcal{X}$ is a measurement schedule on $\ket{G}$.
        \item $\mathcal{X}$ is a path decomposition of $G$.
    \end{enumerate}
\end{theorem}
\begin{proof}
    Let $\mathcal{X}=(X_i)_{i=1}^n$ be a measurement schedule. Condition M1 implies that for each $v \in V$, there exists an $i$ such that $v \in X_i$, which is condition P1. Let $\{u,v\} \in E$ be an edge and let $i$ and $j$ be the maximum indices such that $u \in X_i$ and $v \in X_j$. These indices are guaranteed to exist by condition M1 and we assume without loss of generality that $i \leq j$. By condition M2, it follows that $v \in X_i$, and so $\{u,v\} \subseteq X_i$. Therefore, conditions M1 and M2 imply that for each $e \in E$, there exists an $i$ such that $e \subseteq X_i$, which is condition P2. Conditions M1 and M2 imply that for all $i \leq j \leq k$, if $v \in X_i \cap X_k$, then $v \in X_j$, which is condition P3. Hence, $\mathcal{X}$ is a path decomposition.
    
    Now, let $\mathcal{X}=(X_i)_{i=1}^n$ be a path decomposition. Conditions P1 and P3 imply that a qubit is initialised exactly once, which is M1. Let $v \in V$ be a vertex and let $j$ be the maximum index such that $v \in X_j$, which is guaranteed to exist by condition P1. By condition P2, the neighbourhood of $v$ is contained in $\bigcup_{i=1}^jX_i$. Therefore, conditions P1 and P2 imply that a qubit is measured only after all neighbouring qubits are initialised, which is condition M2. Hence, $\mathcal{X}$ is a measurement schedule. This completes the proof.
\end{proof}
We note this result can be considered an application of the \emph{node searching game} from algorithmic graph theory~\cite{kirousis1985interval}. We have the immediate corollary.
\begin{corollary}
    \label{corollary:SpatialCostPathWidthEquivalence}
    Let $G$ be a graph. An optimal measurement schedule on $\ket{G}$ is determined by a path decomposition of $G$ of minimal width, i.e., $\operatorname{sc}(\ket{G})=\operatorname{pw}(G)+1$.
\end{corollary}
This result shows that the spatial cost of a measurement-based quantum computation scales with the pathwidth of the graph.

\section{Computational Complexity}
\label{section:ComputationalComplexity}

In this section, we explore the computational complexity of the spatial cost and optimal measurement schedules. We first show that approximating the spatial cost is \textsf{NP}-hard. Bodlaender et al.~\cite{bodlaender1995approximating} showed that the problem of approximating the pathwidth up to an additive error is \textsf{NP}-hard. This gives rise to the following corollary.
\begin{corollary}
    Let $G=(V, E)$ be a graph. It is \emph{\textsf{NP}-hard} to approximate $\operatorname{sc}(\ket{G})$ up to an additive error of $\abs{V}^\epsilon$ for $0<\epsilon<1$.
\end{corollary}
\begin{proof}
    The proof follows from Corollary~\ref{corollary:SpatialCostPathWidthEquivalence} and Ref.~\cite[Theorem 23]{bodlaender1995approximating}, which states that it is \textsf{NP}-hard to approximate the pathwidth up to an additive error of $\abs{V}^\epsilon$ for $0<\epsilon<1$.
\end{proof}
While approximating the spatial cost is \textsf{NP}-hard, we shall see there is an fixed-parameter tractable algorithm when parameterised by the spatial cost. This follows from results on the fixed-parameter tractability of the pathwidth~\cite{bodlaender1993linear, furer2016faster}. We obtain the following corollary.
\begin{corollary}
    Let $G=(V, E)$ be a graph. The spatial cost of $\ket{G}$ and a corresponding measurement schedule can be computed in time $\exp[O\left(\operatorname{sc}(\ket{G})^2\right)]\cdot\abs{V}$.
\end{corollary}
\begin{proof}
    The proof follows from Theorem~\ref{theorem:PathDecompositionMeasurementScheduleEquivalence}, Corollary~\ref{corollary:SpatialCostPathWidthEquivalence}, and Ref.~\cite[Theorem 2]{furer2016faster}, which establishes an algorithm for computing the pathwidth and a corresponding path decomposition in time $\exp[O\left(\operatorname{pw}(G)^2\right)]\cdot\abs{V}$. 
\end{proof}
This result implies that an optimal measurement schedule can be efficiently computed for graphs with bounded spatial cost. However, as we shall see, there exists an efficient classical algorithm for simulating measurement-based quantum computation in such cases. Markov and Shi~\cite{markov2008simulating} showed that there is an efficient algorithm for simulating a measurement-based quantum computation for graphs with bounded treewidth. We apply their result to obtain the following corollary.
\begin{corollary}
    Let $G=(V, E)$ be a graph. A measurement-based quantum computation on the graph state $\ket{G}$ can be simulated by a randomised algorithm in time $\exp[O\left(\operatorname{sc}(\ket{G})\right)]\cdot\abs{V}^{O(1)}$.
\end{corollary}
\begin{proof}
    Ref.~\cite[Theorem 6.2]{markov2008simulating} establishes a randomised algorithm for simulating a measurement-based quantum computation on a graph state $\ket{G}$ in time $\exp[O\left(\operatorname{tw}(G)\right)]\cdot\abs{V}^{O(1)}$, where $\operatorname{tw}(G)$ denotes the treewidth of $G$. The proof then follows from Corollary~\ref{corollary:SpatialCostPathWidthEquivalence} and the fact that the treewidth is bounded from above by the pathwidth.
\end{proof}

\section{Implementations}
\label{section:Implementations}

We now explore the implications of our results for implementations of fault-tolerant quantum computing. Markov and Shi~\cite{markov2008simulating} showed that a high treewidth is a necessary condition for a quantum computation to be hard to simulate classically. It follows from Corollary~\ref{corollary:SpatialCostPathWidthEquivalence} that the spatial cost is at least the treewidth plus one. This suggests that the most efficient use of a quantum computer may be realised with graph states in which these two values are equal. Examples of such graphs include the complete graph $K_n$ on $n$ vertices, for which $\operatorname{sc}(\ket{K_n})=\operatorname{tw}(K_n)+1=n$, and the $m \times n$ square lattice $\boxplus_{m,n}$, for which $\operatorname{sc}(\ket{\boxplus_{m,n}})=\operatorname{tw}(\boxplus_{m,n})+1=\min(m,n)+1$.

Quantum devices are constrained by the limitations of their implementation. Therefore, identifying the optimal structure requires considering a specific implementation. Scalable quantum information processing devices are likely to require a significant level of active error correction. By encoding logical information in many physical qubits with an error-correcting code, it is possible to correct errors through measurement~\cite{terhal2015quantum, campbell2017roads}. The most suitable codes are those that are compatible with the constraints of physical devices. The surface code is a prominent choice for such a code, due to its favourable threshold-to-resource ratio and low-weight stabilisers, which make it experimentally feasible~\cite{kitaev2003fault, fowler2012surface}. Although the surface code does not support fault-tolerant application of non-Clifford operations, these operations can be facilitated by preparing magic states in separate magic-state factories and subsequently teleporting them into the qubit register.

We consider an architecture where physical qubits are arranged in a regular two-dimensional lattice. These physical qubits are partitioned into logical surface code patches, quantum bus channels, and magic-state factories~\cite{divincenzo2000physical, litinski2019game}. In the context of measurement-based quantum computation, logical qubits encoded into surface code patches are assigned to qubits of the graph state. The entangling operations used to prepare the graph state are performed using lattice surgery~\cite{herr2017lattice, yoder2017surface, erhard2021entangling}. Measurements involve initially applying logical single-qubit rotations on the surface code patches, followed by Pauli measurements. It is also possible to initialise a logical surface code patch in a magic state prior to entangling the qubit into the graph state.

Graph states with higher-degree vertices necessitate more entangling gates being applied to the same logical qubit. Additionally, non-nearest neighbour entangling gates require quantum bus channels, resulting in an additional space requirement. Note that graphs of degree higher than four also require quantum bus channels. Therefore, a low-degree graph that facilitates only nearest neighbour interactions, such as the square lattice, is a suitable choice for reducing spatial resources.

\begin{figure}[ht!]
    \centering
    \includegraphics[width=0.40\textwidth]{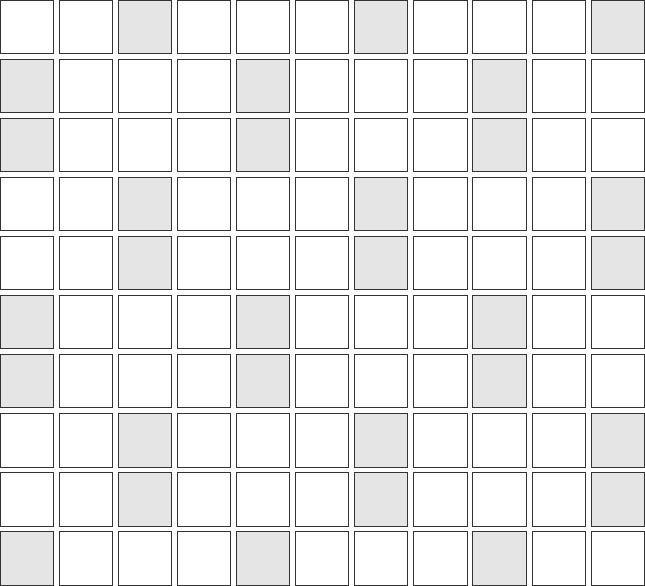}
    \caption{An implementation of the square lattice graph state. The white squares represent logical qubits and the shaded squares represent magic-state factories.}
    \label{figure:squarelatticeimplementation}
\end{figure}

The implementation of the square lattice graph state with a specified measurement schedule allows for the assignment of physical qubits in a two-dimensional array to logical qubits. This approach eliminates the need for quantum bus channels and ensures that each logical qubit is adjacent to a magic-state factory (see Fig.~\ref{figure:squarelatticeimplementation}), thereby optimising the ratio of active logical to physical qubits. Additionally, employing transversal injection methods~\cite{gavriel2023transversal, gidney2023cleaner} for magic-state preparation may further reduce spatial resources compared to distillation methods~\cite{litinski2019magic}.

\section{Conclusion \& Outlook}
\label{section:ConclusionAndOutlook}

We have established an equivalence between the scheduling of graph states in measurement-based quantum computation and path decompositions of graphs. Consequently, we have shown that an optimal measurement schedule is given by a path decomposition of minimal width. Further, we have shown that approximating the spatial cost of a graph is \textsf{NP}-hard, while for graphs with bounded spatial cost, we established an efficient algorithm for computing an optimal measurement schedule. Finally, we discussed the implications of our results for implementations of fault-tolerant quantum computing.

It would be interesting to compare the spatial cost of implementing measurement-based quantum computing in the fault-tolerant setting discussed in Section~\ref{section:Implementations} with an inbuilt error-correcting scheme such as in Refs.~\cite{raussendorf2006fault, raussendorf2007fault}. It would also be interesting to explore the time efficiency of measurement-based quantum computation. While a square lattice may be less temporally efficient than a specifically designed graph state, this does not take into account the costs associated with transporting logical qubits and performing error correction.

\section*{Acknowledgements}

We thank Michael Bremner, Simon Devitt, Salini Karuvade, Thinh Le, Luke Mathieson, and Jannis Ruh for helpful discussions. SJE was supported with funding from the Defense Advanced Research Projects Agency under the Quantum Benchmarking (QB) program under award no. HR00112230007, HR001121S0026, and HR001122C0074 contracts. JG and RLM were supported by the ARC Centre of Excellence for Quantum Computation and Communication Technology (CQC2T), project number CE170100012. The views, opinions and/or findings expressed are those of the authors and should not be interpreted as representing the official views or policies of the Department of Defense or the U.S. Government.

\bibliography{bibliography}

\end{document}